\def\year{2018}\relax
\documentclass[letterpaper]{article} 
\usepackage{aaai18}  
\usepackage{times}  
\usepackage{helvet}  
\usepackage{courier}  
\usepackage{url}  
\usepackage{graphicx}  
\usepackage{booktabs}
\usepackage{amsmath}

\usepackage{multirow}

\usepackage{amsthm}

\usepackage{algorithm}
\usepackage{algorithmic}

\newtheorem{definition}{Definition}
\newtheorem{property}{Property}
\newtheorem{theorem}{Theorem}

\usepackage{graphicx}
\usepackage{caption}

\usepackage{tabularx}

\begin{document}
%
\title{Structural Deep Embedding for Hyper-Networks}

\author{Ke Tu$^{1}$, Peng Cui$^{1}$, Xiao Wang$^{1}$, Fei Wang$^{2}$, Wenwu Zhu$^{1}$\\
	$^{1}$ Department of Computer Science and Technology, Tsinghua University\\
	$^{2}$ Department of Healthcare Policy and Research, Cornell University\\
	\normalsize{tuke1993@gmail.com, cuip@tsinghua.edu.cn, wangxiao007@mail.tsinghua.edu.cn}\\
	\normalsize{feiwang03@gmail.com,wwzhu@tsinghua.edu.cn}\\
}

\maketitle
\begin{abstract}
Network embedding has recently attracted lots of attentions in data mining. Existing network embedding methods mainly focus on networks with pairwise relationships. In real world, however, the relationships among data points could go beyond pairwise, i.e., three or more objects are involved in each relationship represented by a hyperedge, thus forming hyper-networks. These hyper-networks pose great challenges to existing network embedding methods when the hyperedges are indecomposable, that is to say, any subset of nodes in a hyperedge cannot form another hyperedge. These indecomposable hyperedges are especially common in heterogeneous networks. In this paper, we propose a novel {\em Deep Hyper-Network Embedding (DHNE)} model to embed hyper-networks with indecomposable hyperedges. More specifically, we theoretically prove that any linear similarity metric in embedding space commonly used in existing methods cannot maintain the indecomposibility property in hyper-networks, and thus propose a new deep model to realize a non-linear tuplewise similarity function while preserving both local and global proximities in the formed embedding space. We conduct extensive experiments on four different types of hyper-networks, including a GPS network, an online social network, a drug network and a semantic network. The empirical results demonstrate that our method can significantly and consistently outperform the state-of-the-art algorithms.
\end{abstract}

\section{Introduction}
Nowadays, networks are widely used to represent the rich relationships of data objects in various domains, forming social networks, biology networks, brain networks, etc. Many methods are proposed
for network analysis, among which network embedding
methods~\cite{tang2015line,wang2016structural,deng2016discriminative,ou2015non} arouse more and more interests in recent years. Most of the existing network embedding methods are designed for conventional pairwise networks, where each edge links only a pair of nodes. However, in real world applications, the relationships among data objects are much more complicated and they typically go beyond pairwise. For example, John purchasing a shirt with cotton material forms a high-order relationship $\langle$John, shirt, cotton$\rangle$. 
The network capturing those high-order node relationships is usually referred to as a hyper-network.


A typical way to analyze hyper-network is to expand them into conventional pairwise networks and then apply the analytical algorithms developed on pairwise networks. Clique expansion~\cite{sun2008hypergraph} (Figure~\ref{fig:hypergraph} (c)) and star expansion~\cite{agarwal2006higher} (Figure~\ref{fig:hypergraph} (d)) are two representative techniques to achieve such a goal. In clique expansion, each hyperedge is expanded as a clique. In star expansion, a hypergraph is transformed into a bipartite graph where each hyperedge is represented by an instance node which links to the original nodes it contains. 
These methods assume that the hyperedges are {\em decomposable} either explicitly or implicitly. That is to say, if we treat a hyperedge as a set of nodes, then any subset of nodes in this hyperedge can form another hyperedge. In a homogeneous hyper-network, this assumption is reasonable, as the formation of hyperedges are, in most cases, caused by the latent similarity among the involved objects such as common labels. However, when learning the heterogeneous hyper-network embedding, we need to address the following new requirements.

\begin{enumerate}
	\item \textbf{Indecomposablity}: The hyperedges in heterogeneous hyper-networks are usually indecomposable. In this case, a set of nodes in a hyperedge has a strong relationship, while the nodes in its subset does not necessarily have a strong relationship. For example, in the recommendation system with $\langle$user, movie, tag$\rangle$ relationships, the $\langle$user, tag$\rangle$ relationships are not typically strong. This means that we cannot simply decompose hyperedges using those traditional expansion methods.
	
	\item \textbf{Structure Preserving}:
	The local structures are preserved by the observed relationships in network embedding. However, due to the sparsity of networks, many existing relationships are not observed. It is not sufficient for preserving hyper-network structure using only local structures. And global structures, e.g. the neighborhood structure, are required to address the sparsity problem. How to capture and preserve both local and global structures simultaneously in a hyper-network is still an unsolved problem.
\end{enumerate}


To address \textbf{Indecomposablity} issue, we design an indecomposable tuplewise similarity function. The function is directly defined over all the nodes in a hyperedge, ensuring that the subsets of a hyperedge are not incorporated in network embedding. We theoretically prove that the indecomposable tuplewise similarity function can not be a linear function. Therefore, we realize the tuplewise similarity function with a deep neural network and add a non-linear activation function to make it highly non-linear. To address \textbf{Structure Preserving} issue, we design a deep autoencoder to learn node representations by reconstructing neighborhood structures, ensuring that the nodes with similar neighborhood structures will have similar embeddings. The tuplewise similarity function and deep autoencoder are jointly optimized to simultaneously address the two issues.


\begin{figure}
	\includegraphics[width=\linewidth]{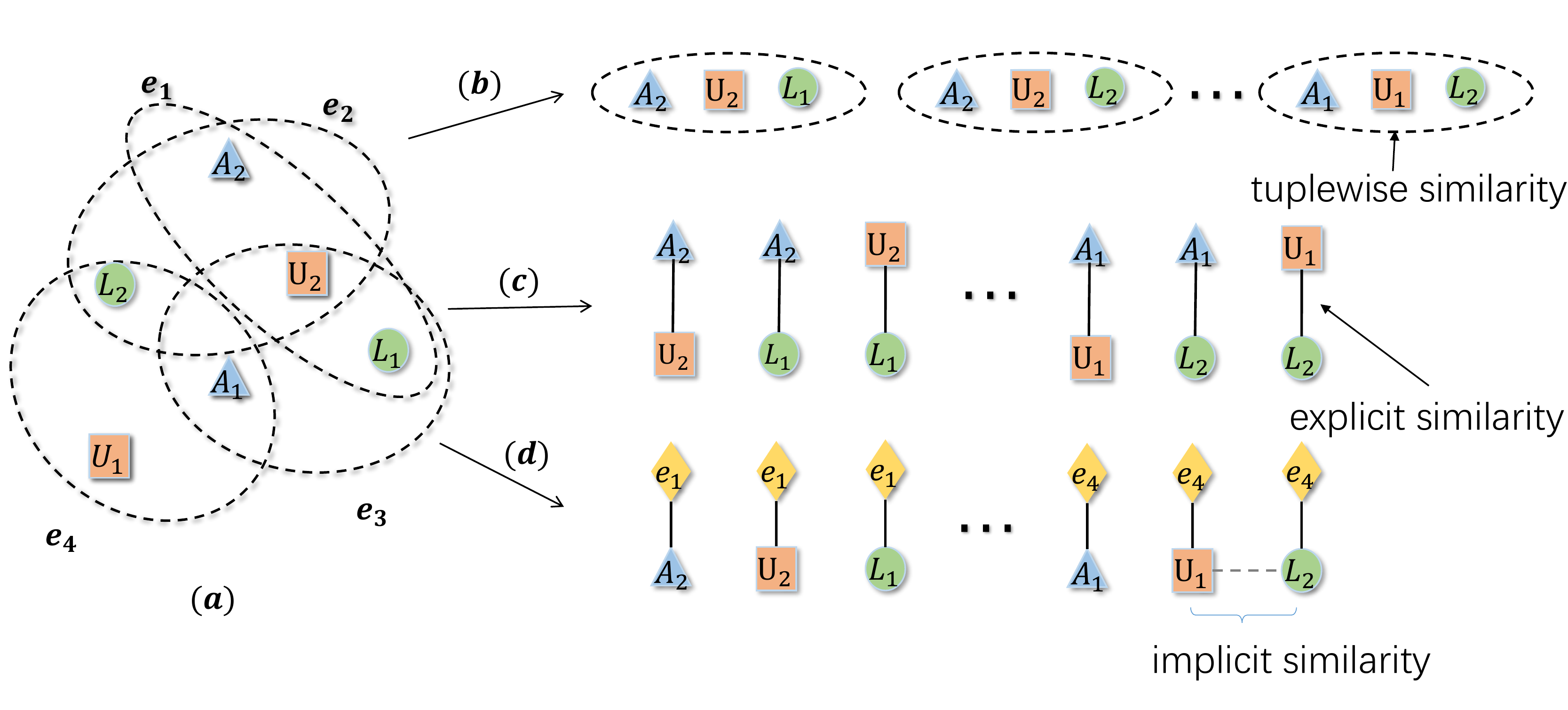}
	\caption{\small (a) An example of a hyper-network. (b) Our method. (c) The clique expansion. (d) The star expansion. Our method models the hyperedge as a whole and the tuplewise similarity is preserved. In clique expansion, each hyperedge is expanded into a clique. Each pair of nodes has explicit similarity. As for the star expansion, each node in one hyperedge links to a new node which stands for the origin hyperedge. Each pair of nodes in the origin hyperedge has implicit similarity for the reason that they link to the same node.}
	\label{fig:hypergraph}
\end{figure}

It is worthwhile to highlight the following contributions of this paper:
\begin{itemize}
	\item We investigate the problem of {\em indecomposable hyper-network embedding}, where indecomposibility of hyperedges is a common property in hyper-networks but largely ignored in literature.
	We propose a novel deep model, named Deep Hyper-Network Embedding (DHNE), to learn embeddings	for the nodes in heterogeneous hyper-networks, which can simultaneously address indecomposable hyperedges while preserving rich structural information. The complexity of this method is linear to the number of node and it can be used in large scale networks
	\item We theoretically prove that any linear similarity metric in embedding space cannot maintain the indecomposibility property in hyper-networks, and thus propose a novel deep model to simultaneously maintain the indecomposibility as well as the local and global structural information in hyper-networks.
	\item We conduct experiments on four real-world information
	networks. The results demonstrate the effectiveness and efficiency of the proposed model.
\end{itemize}

The remainder of this paper is organized as follows. In the next
section, we review the related work. Section 3 gives the preliminaries and formally defines our problem. In Section 4,  we introduce the proposed model in details. Experimental results are presented in section 5. Finally, we conclude in Section 6.
\section{Related work}
Our work is related to network embedding which aims to learn low-dimension representations for networks. Earlier works, such as Local Linear Embedding (LLE)~\cite{roweis2000nonlinear}, Laplacian eigenmaps~\cite{belkin2001laplacian} and IsoMap~\cite{tenenbaum2000global}, are based on matrix factorization. They express a network as a matrix where the entries represent relationships and calculate the leading eigenvectors as network representations. Eigendecomposition is a very expensive operation so these methods cannot efficiently scale to large real world networks. Recently, DeepWalk~\cite{perozzi2014deepwalk} learns latent representations of nodes in a network by modeling a stream of short random walks. LINE~\cite{tang2015line} optimizes an objective function which aims to preserve both the first-order and second-order proximities of networks. HOPE~\cite{ou2016asymmetric} extends the work to utilize higher-order information and M-NMF~\cite{wang2017community} incorporates the community structure into network embedding. Furthermore, due to the powerful representation ability of deep learning~\cite{niepert2016learning}, several network embedding methods based on deep learning ~\cite{chang2015heterogeneous,wang2016structural} have been proposed. ~\cite{wang2016structural} proposes a deep model with a semi-supervised architecture, which simultaneously optimizes the first-order and second-order proximity. ~\cite{chang2015heterogeneous} employs deep model to transfer different objects in heterogeneous networks to unified vector representations. 

However, all of the above methods assume pairwise relationships among objects in real world networks. In view of the aforementioned facts, a series of methods~\cite{zhou2006learning,liu2013hypergraph,wu2010multiple} is proposed by generalizing spectral clustering techniques~\cite{ng2001spectral} to hypergraphs. Nevertheless, these methods focus on homogeneous hypergraph. They construct hyperedge by latent similarity like common label and preserve hyperedge implicitly. Therefore, they cannot preserve the structure of indecomposable hyperedges. For heterogeneous hyper-network, the tensor decomposition~\cite{kolda2009tensor,rendle2010pairwise,symeonidis2016matrix} may be directly applied to learn the embedding. Unfortunately, the time cost of tensor decomposition is usually very expensive so it cannot scale efficiently to large network. Besides,  HyperEdge Based Embedding (HEBE)~\cite{gui2016large} is proposed to model the proximity among participating objects in each heterogeneous event as a hyperedge based on prediction. It does not take high-order network structure and high degree of sparsity into account which affects the predictive performance in our task.

\section{Notations and Definitions}
In this section, we define the problem of hyper-network embedding. The key notations used in this paper are shown in Table~\ref{tab:notations}. First we give the definition of a hyper-network.

\begin{table}
	\caption{\small Notations.}
	\label{tab:notations}
	\resizebox{\linewidth}{!}{
	\begin{tabular}{|c|c|}
		\toprule
		Symbols & Meaning \\
		\midrule
		$T$ & number of node types \\
		$\mathbf{V} = \{\mathbf{V}_t\}_{t=1}^T$ & node set \\
		$\mathbf{E} = \{(\emph{v}_1, \emph{v}_2, ..., \emph{v}_{n_i})\}$ & hyperedge set \\ 
		$\mathbf{A}$ &  adjacency matrix of hyper-network \\
		$\mathbf{X}_i^{\emph{j}}$ &  embedding of node $i$ with type $\emph{j}$ \\
		$\mathcal{S}(\mathbf{X}_1, \mathbf{X}_2,..., \mathbf{X}_N)$ & $N$-tuplewise similarity function \\
		$\mathbf{W}^{(i)}_j$ & the i-th layer weight matrix with type j \\
		$\mathbf{b}^{(i)}_j$ & the i-th layer biases with type j \\
		\bottomrule
	\end{tabular}}
\end{table}

\begin{definition}[Hyper-network]
	A \textbf{hyper-network} is defined as a hypergraph $\mathbf{G} = (\mathbf{V}, \mathbf{E})$ with the set of nodes $V$ belonging to $T$ types $\mathbf{V} = \{\mathbf{V}_t\}_{t=1}^T$ and the set of edges $\mathbf{E}$ which may have more than two nodes $\mathbf{E} = \{\emph{E}_i=(\emph{v}_1, \emph{v}_2, ..., \emph{v}_{n_i})\} (n_i \geq 2)$. If the number of nodes is 2 for each hyperedge, the hyper-network degenerates to a network. The type of edge $\emph{E}_i$ is defined as the combination of types of nodes belonging to the edge. If $T \geq 2$, the hyper-network is defined as a \textbf{heterogeneous hyper-network}.
\end{definition}

To obtain the embeddings in a hyper-network, the indecomposable tuplewise relationships need be preserved. We define the indecomposable structures as the first-order proximity of hyper-network:

\begin{definition}[The First-order Proximity of Hyper-network]
	\label{def:first-order}
	\textbf{The first-order proximity} of hyper-network measures the 
	N-tuplewise similarity between nodes. For any $N$ vertexes $\emph{v}_1, \emph{v}_2, ..., \emph{v}_{N}$, if there exists a hyperedge among these $N$ vertexes, the first-order proximity of these $N$ vertexes is defined as 1, but this implies no first-order proximity for any subsets of these N vertexes.
	
\end{definition}

The first-order proximity implies the indecomposable similarity of several entities in real world. Meanwhile, real world networks are always incomplete and sparse. Only considering first-order proximity is not sufficient for learning node embeddings. Higher order proximity needs to be considered to fix this issue. We then introduce the second-order proximity of hyper-network to capture the global structure.

\begin{definition}[The Second-order Proximity of Hyper-network]
	\textbf{The second-order Proximity} of hyper-network measures the proximity of two nodes with respect to their neighborhood structures. For any node $\emph{v}_i \in \emph{E}_i$, $\emph{E}_i/{\emph{v}_i}$ is defined as a \textbf{neighborhood} of $\emph{v}_i$. If $\emph{v}_i$'s neighborhoods $\{\emph{E}_i/{\emph{v}_i} ~for~any~\emph{v}_i \in \emph{E}_i\}$ are similar to $\emph{v}_j$'s, then $\emph{v}_i$'s embedding $\mathbf{x}_i$ should be similar to $\emph{v}_j$'s embedding $\mathbf{x}_j$.
\end{definition}

For example, in Figure \ref{fig:hypergraph}(a), $\emph{A}_1$'s neighborhoods set is $\{(\emph{L}_2, \emph{U}_1), (\emph{L}_1, \emph{U}_2)\}$. $\emph{A}_1$ and $\emph{A}_2$ have second-order similarity since they have common neighborhood, $(\emph{L}_1, \emph{U}_2)$.

\section{Deep Hyper-Network Embedding}

In this section, we introduce the proposed Deep Hyper-Network Embedding (DHNE). The framework is shown in Figure \ref{fig:framework}.

\subsection{Loss function}
To preserve the first-order proximity of a hyper-network, an $N$-tuplewise similarity measure in embedding space is required. If there exists a hyperedge among $N$ vertexes, the $N$-tuplewise similarity of these vertexes should be large, and small otherwise.

\begin{property}
	\label{property}
	We mark $\mathbf{X}_i$ as the embedding of node $v_i$ and $\mathcal{S}$ as $N$-tuplewise similarity function.
	\begin{itemize}
		\item if $(\emph{v}_1, \emph{v}_2, ..., \emph{v}_N) \in \mathbf{E}$, $\mathcal{S}(\mathbf{X}_1, \mathbf{X}_2, .., \mathbf{X}_N)$ should be large (without loss of generality, large than a threshold $l$). 
		\item if $(\emph{v}_1, \emph{v}_2, ..., \emph{v}_N) \notin \mathbf{E}$, $\mathcal{S}(\mathbf{X}_1, \mathbf{X}_2, .., \mathbf{X}_N)$ should be small (without loss of generality, smaller than a threshold $s$).
	\end{itemize}
\end{property}

\begin{figure} 
	\includegraphics[width=\linewidth]{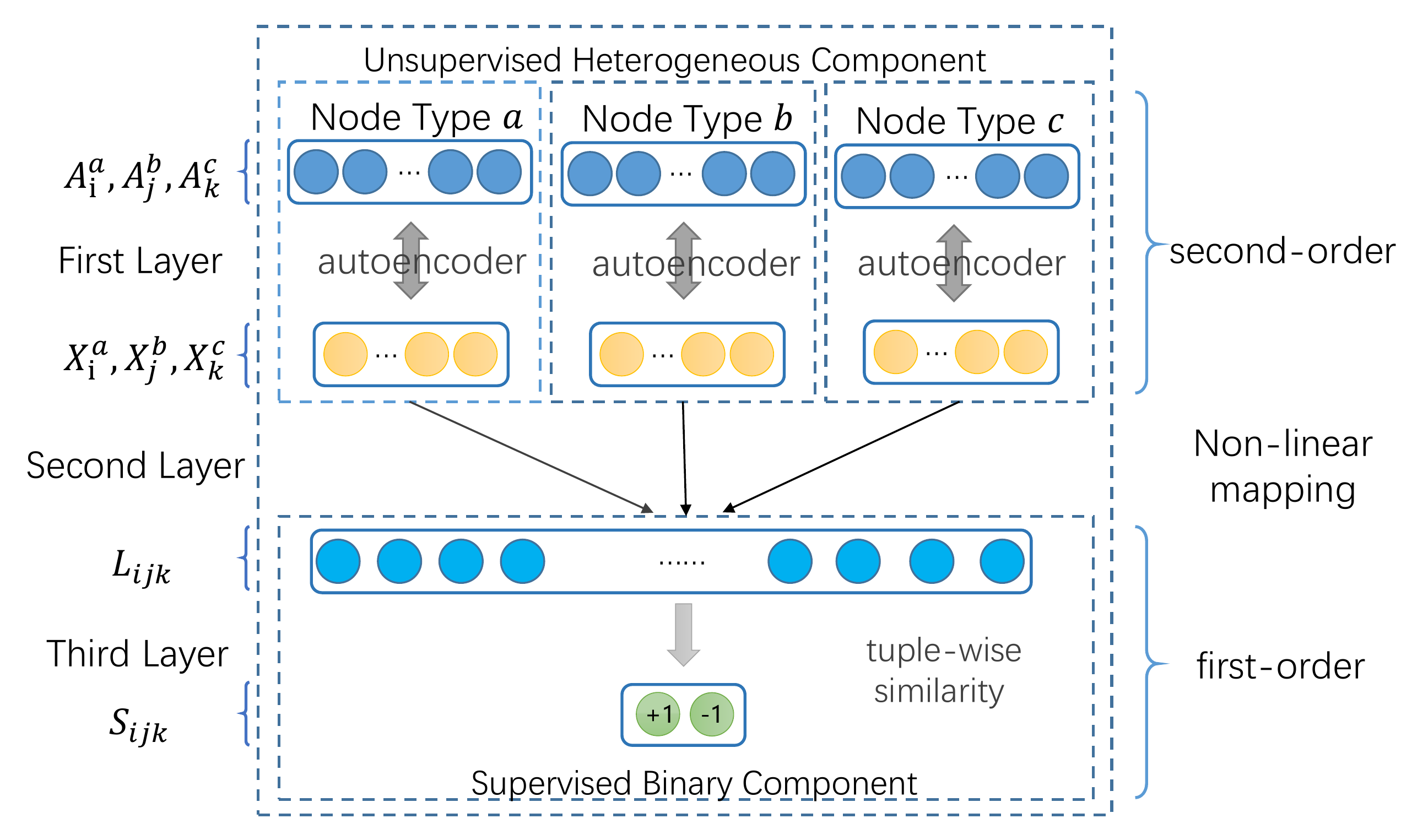}
	\caption{{\small Framework of Deep Hyper-Network Embedding.}}
	\label{fig:framework}
\end{figure}

In our model, we propose a data-dependent $N$-tuplewise similarity function. In this paper, we mainly focus on hyperedges with uniform length $N = 3$, but it is easy to extend to $N > 3$.

Here we provide the theorem to demonstrate that a linear tuplewise similarity function cannot satisfy Property \ref{property}.

\begin{theorem}
	Linear function $\mathcal{S}(\mathbf{X}_1, \mathbf{X}_2,..., \mathbf{X}_N) = \sum_i \mathbf{W}_i\mathbf{X}_i$ cannot satisfy Property \ref{property}.
	\label{theorem}
\end{theorem}

\begin{proof}
	To prove it by contradiction, we assume that theorem~\ref{theorem} is false, i.e., the linear function $\mathcal{S}$ satisfies Property~\ref{property}. We suggest the following counter example. Assume we have 3 types of nodes, and each type of node has 2 clusters (0 and 1). There is a hyperedge if and only if 3 nodes from different types have the same cluster id. We use $\mathbf{Y}_i^{j}$ to represent embeddings of nodes with type $j$ in cluster $i$. By Property \ref{property}, we have
	\begin{small}
	\begin{align}
		\mathbf{W}_1\mathbf{Y}_0^1+\mathbf{W}_2\mathbf{Y}_0^2+\mathbf{W}_3\mathbf{Y}_0^3 > l \label{equ:ce1}&\\
		\mathbf{W}_1\mathbf{Y}_1^1+\mathbf{W}_2\mathbf{Y}_0^2+\mathbf{W}_3\mathbf{Y}_0^3 < s \label{equ:ce2}&\\
		\mathbf{W}_1\mathbf{Y}_1^1+\mathbf{W}_2\mathbf{Y}_1^2+\mathbf{W}_3\mathbf{Y}_1^3 > l \label{equ:ce3}&\\
		\mathbf{W}_1\mathbf{Y}_0^1+\mathbf{W}_2\mathbf{Y}_1^2+\mathbf{W}_3\mathbf{Y}_1^3 < s \label{equ:ce4}&.
	\end{align}
	\end{small}
	By combining Equation (\ref{equ:ce1})(\ref{equ:ce2})(\ref{equ:ce3})(\ref{equ:ce4}), we get $\mathbf{W}_1*(\mathbf{Y}_0^1-\mathbf{Y}_1^1) > l-s$ and $\mathbf{W}_1*(\mathbf{Y}_1^1-\mathbf{Y}_0^1) > l-s$, which is a contradiction.
	
	This completes the proof.
\end{proof}

Above theorem shows that N-tuplewise similarity function $\mathcal{S}$ should be in a non-linear form. This motivates us to model it by a multilayer perceptron. The multilayer perceptron is composed of two parts, which are shown separately in the second layer and third layer of Figure \ref{fig:framework}. The second layer is a fully connected layer with non-linear activation functions. With the input of the embeddings $(\mathbf{X}_i^\emph{a}, \mathbf{X}_j^\emph{b}, \mathbf{X}_k^\emph{c})$ of 3 nodes $(\emph{v}_i, \emph{v}_j, \emph{v}_k)$, we concatenate them and map them non-linearly to a common latent space $\mathbf{L}$. Their joint representation in latent space is shown as follows:

\begin{equation}
\mathbf{L}_{ijk} = \sigma(\mathbf{W}^{(2)}_\emph{a}*\mathbf{X}_i^\emph{a}+
\mathbf{W}^{(2)}_\emph{b}*\mathbf{X}_j^\emph{b}+
\mathbf{W}^{(2)}_\emph{c}*\mathbf{X}_k^\emph{c}+\mathbf{b}^{(2)}), 
\end{equation}
where $\sigma$ is the sigmoid function.

After obtaining the latent representation $\mathbf{L}_{ijk}$, we finally map it to a probability space in the third layer to get the similarity:

\begin{equation}
\mathbf{S}_{ijk} \equiv \mathcal{S}(\mathbf{X}^\emph{a}_i, \mathbf{X}^\emph{b}_j, \mathbf{X}^\emph{c}_k) = \sigma(\mathbf{W}^{(3)}*\mathbf{L}_{ijk}+\mathbf{b}^{(3)}).
\end{equation}

Combining the aforementioned two layers, we obtain a non-linear tuplewise similarity measure function $\mathcal{S}$. In order to make this similarity function satisfy Property \ref{property}, we present the objective function as follows:

\begin{equation}
\label{equ:first-order}
\mathcal{L}_1 = -(\mathbf{R}_{ijk}\log \mathbf{S}_{ijk}+(1-\mathbf{R}_{ijk})\log(1-\mathbf{S}_{ijk})),
\end{equation}
where $\mathbf{R}_{ijk}$ is defined as $1$ if there is a hyperedge between $\emph{v}_i$, $\emph{v}_j$ and $\emph{v}_k$ and $0$ otherwise. From the objective function, it is easy to check that if $\mathbf{R}_{ijk}$ equals to $1$, the similarity $\mathbf{S}_{ijk}$ should be large, and otherwise the similarity should be small. In other words, the first-order proximity is preserved. 

Next, we consider to preserve the second-order proximity. The first layer of Figure \ref{fig:framework} is designed to preserve the second-order proximity. Second-order proximity measures neighborhood structure similarity. Here, we define the adjacency matrix of hyper-network to capture the neighborhood structure. First, we give some basic definitions of hypergraph. For a hypergraph $\mathbf{G} = (\mathbf{V}, \mathbf{E})$, a $|\mathbf{V}|*|\mathbf{E}|$ incidence matrix $\mathbf{H}$ with entries $\mathbf{h}(\emph{v}, \emph{e}) = 1$ if $\emph{v} \in \emph{e}$ and 0 otherwise, is defined to represent the hypergraph. For a vertex $\emph{v} \in \mathbf{V}$, the degree of vertex is defined by $d(v) = \sum_{\emph{e} \in \mathbf{E}} \mathbf{h}(\emph{v}, \emph{e})$. Let $\mathbf{D}_v$ denote the diagonal matrix containing the vertex degree. Then the adjacency matrix $\mathbf{A}$ of hypergraph $\mathbf{G}$ can be defined as $\mathbf{A} = \mathbf{H}\mathbf{H}^T-\mathbf{D}_v$, where $\mathbf{H}^T$ is the transpose of $\mathbf{H}$. The entries of adjacency matrix $\mathbf{A}$ denote the concurrent times between two nodes, and the i-th row of adjacency matrix $\mathbf{A}$ shows the neighborhood structure of vertex $\emph{v}_i$. We use an adjacency matrix $\mathbf{A}$ as our input feature and an autoencoder~\cite{lecun2015deep} as the model to preserve the neighborhood structure. The autoencoder is composed by an encoder and a decoder. The encoder is a non-linear mapping from feature space $\mathbf{A}$ to latent representation space $\mathbf{X}$ and the decoder is a non-linear mapping from latent representation $\mathbf{X}$ space back to origin feature space $\hat{\mathbf{A}}$, which is shown as follows:
\begin{align}
	\mathbf{X}_i &= \sigma(\mathbf{W}^{(1)}*\mathbf{A}_i+\mathbf{b}^{(1)}) \label{equ:encoder} \\ 
	\hat{\mathbf{A}}_i &= \sigma(\hat{\mathbf{W}}^{(1)}*\mathbf{X}_i+\hat{\mathbf{b}}^{(1)}). \label{equ:decoder}
\end{align}

The goal of autoencoder is to minimize the reconstruction error between the input and the output. The autoencoder's reconstruction process will make the nodes with similar neighborhoods have similar latent representations, and thus the second-order proximity is preserved. It is noteworthy that the input feature is the adjacency matrix of the hyper-network, and the adjacency matrix is often extremely sparse. To speed up our model, we only reconstruct non-zero element in the adjacency matrix. The reconstruction error is shown as follows:
\begin{equation}
	||sign(\mathbf{A}_i)\odot(\mathbf{A}_i - \hat{\mathbf{A}}_i)||_F^2,
\end{equation}
where $sign$ is the sign function.

Furthermore, in hyper-networks, the vertexes often have various types, forming heterogeneous hyper-networks. Considering the special characteristics of different types of nodes, it is required to learn unique latent spaces for different node types. In our model, each heterogeneous type of entities have their own autoencoder model as shown in Figure \ref{fig:framework}. Then for all types of nodes, the loss function is defined as:

\begin{equation}
\label{equ:second-order}
\mathcal{L}_2 = \sum_{t} ||sign(\mathbf{A}_i^t)\odot(\mathbf{A}_i^t - \hat{\mathbf{A}}_i^t)||_F^2,
\end{equation}
where t is the index for node types.

To preserve both first-order proximity and second-order proximity of heterogeneous hyper-networks, we jointly minimize the objective function by combining Equation \ref{equ:first-order} and Equation \ref{equ:second-order}:

\begin{equation}
\label{equ:obj}
\mathcal{L} = \mathcal{L}_1+\alpha\mathcal{L}_2.
\end{equation}

\subsection{Optimization}
We use stochastic gradient descent (SGD) to optimize the model. The key step is to calculate the partial derivative of the parameters $\theta = \{\mathbf{W}^{(i)}, \mathbf{b}^{(i)}, \mathbf{\hat{W}}^{(i)}, \mathbf{\hat{b}}^{(i)}\}_{i=1}^3$. These derivatives can be easily estimated by using back-propagation algorithm~\cite{lecun2015deep}. Notice that there is only positive relationship in most real world network, so this algorithm may converge to a trivial solution where all tuplewise relationships are similar. To address this problem, we sample multiple negative edges based on noisy distribution for each edge, as in ~\cite{mikolov2013distributed}. The whole algorithm is shown in Algorithm~\ref{algorithm}.

\begin{algorithm}
	\caption{\small The Deep Hyper-Network Embedding (DHNE)}
	\label{algorithm}
	\begin{algorithmic}[1]
		\REQUIRE the hyper-network $\mathbf{G} = (\mathbf{V}, \mathbf{E})$ with adjacency matrix $\mathbf{A}$, the parameter $\alpha$
		\ENSURE Hyper-network Embeddings $E$ and updated Parameters $\theta = \{\mathbf{W}^{(i)}, \mathbf{b}^{(i)}, \mathbf{\hat{W}}^{(i)}, \mathbf{\hat{b}}^{(i)}\}_{i=1}^3$
		\STATE initial parameters $\theta$ by random process
		\WHILE{the value of objective function do not converge} 
		\STATE generate next batch from the hyperedge set $\mathbf{E}$
		\STATE sample negative hyperedge randomly
		\STATE calculate partial derivative $\partial\mathcal{L}/\partial \theta$ with back-propagation algorithm to update $\theta$.
		\ENDWHILE 
	\end{algorithmic}
\end{algorithm}

\subsection{Analysis}
In this section, we present the out-of-sample extension and complexity analysis.

\subsubsection{Out-of-sample extension}
For a newly arrived vertex $\emph{v}$, we can easily obtain its adjacency vector by its connections to existing vertexes. We feed its adjacency vector into the specific autoencoder corresponding to its type, and apply Equation \ref{equ:encoder} to get representation for vertex $v$. The complexity for such steps is $\mathcal{O}(d_vd)$, where $d_v$ is the degree of vertex $\emph{v}$ and $d$ is the dimensionality of the embedding space..
\subsubsection{Complexity analysis}
During the training procedure, the time complexity of calculating gradients and updating parameters is $O((nd+dl+l)bI)$, where $n$ is the number of nodes, $d$ is the dimension of embedding vectors, $l$ is the size of latent layer, $b$ is the batch size and $I$ is the number of iterations. Parameter $l$ is usually related to the dimension of embedding vectors $d$ but independent with the number of vertexes $n$. The batch size is normally a small number. The number of iterations is also not related with the number of vertexes $n$. Therefore, the complexity of training procedure is linear to the number of vertexes.

\section{Experiment}
In this section, we evaluate our proposed method  on several real world datasets and multiple application scenarios.
\subsection{Datasets}
In order to comprehensively evaluate the effectiveness of our proposed method, we use four different types of datasets, including a GPS network, a social network, a medicine network and a semantic network. The detailed information is shown as follows.
\begin{itemize}
	\item GPS~\cite{zheng2010collaborative}: The dataset describes a user  joins in an activity in certain location. The (user, location, activity) relations are used for building the hyper-network.
	\item MovieLens~\cite{harper2016movielens}: This dataset describes personal tagging activity from MovieLens\footnote{https://movielens.org/}. Each movie is labeled by at least one genres. The (user, movie, tag) relations are considered as the hyperedges to form a hyper-network.
	\item drug\footnote{http://www.fda.gov/Drugs/}: This dataset is obtained from FDA Adverse Event Reporting System (FAERS). It contains information on adverse event and medication error reports submitted to FDA. We construct hyper-network by (user, drug, reaction) relationships, i.e., a user who has certain reaction  and takes some drugs will lead to adverse event.
	\item wordnet~\cite{bordes2013translating}: This dataset consists of a collection of triplets (synset, relation type, synset) extracted from WordNet 3.0. We can construct the hyper-network by regarding head entity, relation, tail entity as three types of nodes and the triplet relationships as hyperedges.
\end{itemize}
The detailed statistics of the datasets are summarized in Table \ref{tab:datasets}.

\begin{table}
	\caption{\small Statistics of the datasets.}
	\label{tab:datasets}
	\resizebox{\linewidth}{!}{
			\begin{tabular}{|c|c|c|c|c|c|c|c|}
				\toprule
				datasets & \multicolumn{3}{c|}{node type} & \multicolumn{3}{c|}{\#(V)} & \#(E) \\
				\midrule
				GPS & user & location & activity & 146 & 70 & 5 & 1436 \\
				MovieLens & user & movie & tag & 2113 & 5908 & 9079 & 47957 \\
				drug & user & drug & reaction & 12 & 1076 & 6398 & 171756 \\
				wordnet & head & relation & tail & 40504 & 18 & 40551 & 145966 \\
				\bottomrule
		\end{tabular} }
\end{table}

\subsection{Parameter Settings}
We compared DHNE against the following six widely-used algorithms: DeepWalk~\cite{perozzi2014deepwalk}, LINE~\cite{tang2015line}, node2vec~\cite{grover2016node2vec}, Spectral Hypergraph Embedding (SHE)~\cite{zhou2006learning}, Tensor decomposition~\cite{kolda2009tensor} and HyperEdge Based Embedding (HEBE)~\cite{gui2016large}.

In summary, DeepWalk, LINE and node2vec are conventional pairwise network embedding methods. In our experiment, we use clique expansion in Figure \ref{fig:hypergraph}(c) to transform a hyper-network in a conventional network, and then use these three methods to learn node embeddings from the conventional networks. SHE is designed for homogeneous hyper-network embeddings. Tensor method is a direct way for preserving high-order relationship in heterogeneous hyper-network. HEBE learns node embeddings for heterogeneous event data. Note that DeepWalk, LINE , node2vec and SHE can only measure pairwise relationship. In order to make them applicable to network reconstruction and link prediction in hyper-networks, without loss of generality, we use the mean or minimum value among all pairwise similarities in a candidate hyperedge to represent the tuplewise similarity of the hyperedge. For DeepWalk and node2vec, we set window size as 10, walk length as 40, walks per vertex as 10. For LINE, we set the number of negative samples as 5.

We uniformly set the representation size as 64 for all methods. Specifically, for DeepWalk and node2vec, we set window size as 10, walk length as 40, walks per vertex as 10. For LINE, we set the number of negative samples as 5.

For our model, we use one-layer autoencoder to preserve hyper-network structure and one-layer fully connected layer to learn tuplewise similarity function. The size of hidden layer of autoencoder is set as 64 which is also the representation size.  The size of fully connect layer is set as sum of the embedding length from all types, 192. We do grid search from $\{0.01, 0.1, 1, 2, 5, 10\}$ to tune the parameter $\alpha$ which is shown in {\em Parameter Sensitivity} section. Similar to LINE~\cite{tang2015line}, the learning rate is set with the starting value $\rho_0 = 0.025$ and decreased linearly with the times of iterations.

\begin{table}
	\caption{\small AUC value for network reconstruction.}
	\label{tab:reconstrction_AUC}
	\resizebox{\linewidth}{!}{
	\begin{tabular}{|c|c|c|c|c|c|}
		\toprule
		\multicolumn{2}{|c|}{methods} & GPS & MovieLens & drug & wordnet \\
		\midrule
		\multicolumn{2}{|c|}{DHNE} & \textbf{0.9598} & \textbf{0.9344} & \textbf{0.9356} & \textbf{0.9073} \\
		\midrule
		\multirow{4}{*}{mean} & deepwalk & 0.6714 & 0.8233 & 0.5750 & 0.8176 \\
		& line & 0.8058 & 0.8431 & 0.6908 & 0.8365 \\
		& node2vec & 0.6715 & 0.9142 & 0.6694 & 0.8609 \\
		& SHE & 0.8596 & 0.7530 & 0.5486 & 0.5618 \\
		\midrule
		\multirow{3}{*}{min} & deepwalk & 0.6034 & 0.7117 & 0.5321 & 0.7423 \\
		& line & 0.7369 & 0.7910 & 0.7625 & 0.7751 \\
		& node2vec & 0.6578  & 0.9100 & 0.6557  & 0.8387 \\
		& SHE & 0.7981 & 0.7972 & 0.6236 & 0.5918 \\
		\midrule
		\multicolumn{2}{|c|}{tensor} & 0.9229 & 0.8640 & 0.7025 & 0.7771 \\
		\multicolumn{2}{|c|}{HEBE} & 0.9337 & 0.8772 &  0.8236 & 0.7391 \\
		\bottomrule
	\end{tabular}}
\end{table}

\subsection{Network Reconstruction}
A good network embedding method should preserve the original network structure well in the embedding space. We first evaluate our proposed algorithm on network reconstruction task. We use the learned embeddings to predict the links of origin networks. The AUC (Area Under the Curve)~\cite{fawcett2006introduction} is used as the evaluation metric. 
The results are show in Table \ref{tab:reconstrction_AUC}.

From the results, we have the following observations:
\begin{itemize}
	\item Our method achieves significant improvements on AUC values over the baselines on all four datasets. It demonstrates that our method is able to preserve the origin network structure well.
	\item Compared with the baselines, our method achieves higher improvements in sparse drug and wordnet datasets than those on GPS and MovieLens datasets. It indicates the robustness of proposed	methods on sparse datasets.
	\item The results of DHNE perform better than DeepWalk, LINE and SHE which assume that high-order relationships are decomposable. It demonstrates the importance of preserving indecomposable hyperedge in hyper-network.
\end{itemize}

\subsection{Link Prediction}
Link prediction is a widely-used application in real world especially in recommendation systems. In this section, we fulfil two link prediction tasks on all the four datasets. The two tasks evaluate the overall performance and the performance with different sparsity of the networks, respectively. We calculate AUC value as in network reconstruction task to evaluate the performance.

For the first task, we randomly hide 20 percentage of existing edges and use the left network to train hyper-network embedding models. After training, we obtain embedding for each node and similarity function for N nodes and apply the similarity function to predict the held-out links. For GPS dataset which is a small and dense dataset, we can draw ROC curve for this task to observe the performance at various threshold settings, as shown in Figure \ref{fig:link_prediction_all} left. The AUC results on all datasets are shown in Table \ref{tab:link_prediction_AUC}. The observations are illustrated as follows:
\begin{itemize}
	\item Our method achieves significant improvements over the baselines on all the datasets. It demonstrates the learned embeddings of our method have strong predictive power for unseen links.
	\item By comparing the performance of LINE, DeepWalk, SHE and DHNE, we can observe that transforming the indecomposable high-order relationship into multiple pairwise relationship will damage the predictive power of the learned embeddings.
	\item As Tensor and HEBE can somewhat address the indecomposibility of hyperedges, the large improvement margin of our method over these two methods clearly demonstrates the importance of second-order proximities in hyper-network embedding.
\end{itemize}

For the second task, we change the sparsity of network by randomly hiding different ratios of existing edges and repeat the previous task. Particularly, we conduct this task on the drug dataset as it has the most hyperedges. The ratio of remained edges is selected from 10\% to 90\%. The results are shown in Figure \ref{fig:link_prediction_all} right.


\begin{figure}
	\includegraphics[width=1.0\linewidth]{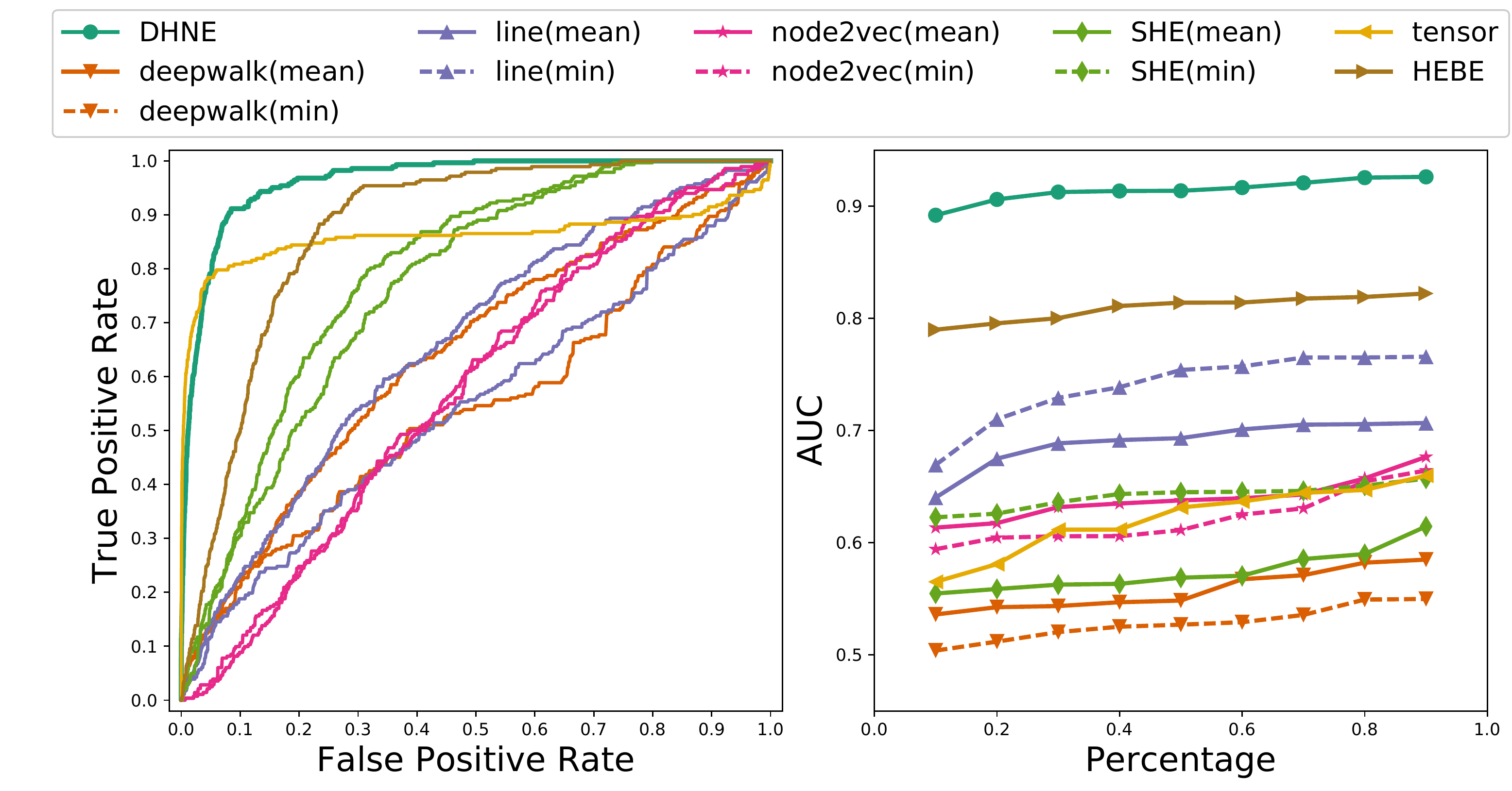}
	\caption{\small left: ROC curve on GPS; right: Performance for link prediction on networks of different sparsity.}
	\label{fig:link_prediction_all}
\end{figure}

We can observe that DHNE has a significant improvements over the best baselines on all sparsity of networks. It demonstrates the effectiveness of DHNE on sparse networks.


\begin{table}
	\caption{\small AUC value for link prediction.}
	\label{tab:link_prediction_AUC}
	\resizebox{\linewidth}{!}{
	\begin{tabular}{|c|c|c|c|c|c|}
		\toprule
		\multicolumn{2}{|c|}{methods} & GPS & MovieLens & drug & wordnet \\
		\midrule
		\multicolumn{2}{|c|}{DHNE} & \textbf{0.9166} & \textbf{0.8676} & \textbf{0.9254} & \textbf{0.8268} \\
		\midrule
		\multirow{4}{*}{mean} & deepwalk & 0.6593 & 0.7151 & 0.5822 & 0.5952 \\
		& line & 0.7795 & 0.7170 & 0.7057 & 0.6819 \\
		& node2vec & 0.5835 & 0.8211　& 0.6573 & 0.8003 \\
		& SHE & 0.8687 & 0.7459 & 0.5899 & 0.5426 \\
		\midrule
		\multirow{4}{*}{min} & deepwalk & 0.5715 & 0.6307 & 0.5493 & 0.5542 \\
		& line & 0.7219 & 0.6265 &  0.7651 & 0.6225 \\
		& node2vec & 0.5869 &　0.7675 & 0.6546 & 0.7985 \\
		& SHE & 0.8078 & 0.8012 & 0.6508 & 0.5507 \\
		\midrule
		\multicolumn{2}{|c|}{tensor} & 0.8646 & 0.7201 & 0.6470 & 0.6516 \\
		\multicolumn{2}{|c|}{HEBE} & 0.8355 & 0.7740 &  0.8191 & 0.6364 \\
		\bottomrule
	\end{tabular}}
	
\end{table}

\subsection{Classification}
In this section, we conduct multi-label classification~\cite{bhatia2015sparse} on MovieLens dataset and multi-class classification in wordnet, because only these two datasets have label or category information. After deriving the node embeddings from different methods, we choose SVM as the classifier. For MovieLens dataset, we randomly sample 10\% to 90\% of the vertexes as the training samples and use the left vertexes to test the performance. For wordnet dataset, the portion of training data is selected from 1\% to 10\%. Besides, we remove the nodes without labels on these two datasets. Averaged Macro-F1 and Micro-F1 are used to evaluate the performance. The results are shown in Figure \ref{fig:classification}.

\begin{figure}
	\includegraphics[width=1.0\linewidth]{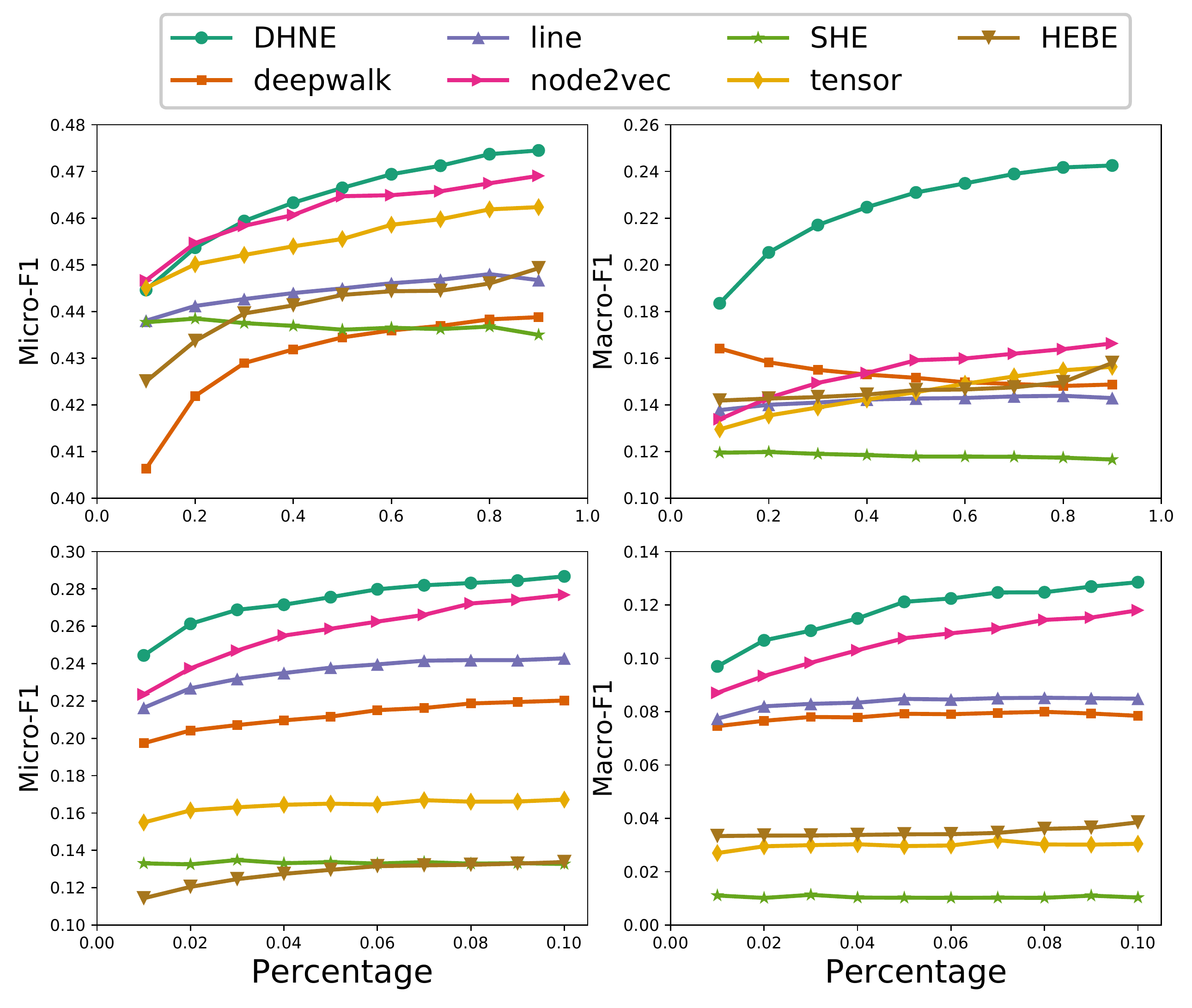}
	\caption{\small top: multi-label classification on MovieLens dataset; bottom: multi-class classification on wordnet dataset.}
	\label{fig:classification}
\end{figure}

%

From the results, we have following observations:
\begin{itemize}
	\item In both Micro-F1 and Macro-F1 curves, our method performs consistently better than baselines. It demonstrates the effectiveness of our proposed method in classification tasks.
	\item When the labelled data becomes richer, the relative improvement of our method is more obvious than baselines. Besides, as shown in Figure \ref{fig:classification} bottom, when the labeled data is quite sparse, our method still outperforms the baselines. This demonstrates the robustness of our method.
\end{itemize}

\subsection{Parameter Sensitivity}
\label{sec:ps}
In this section, we investigate how parameters influence the performance and training time. Especially, we evaluate the effect of the ratio of first-order proximity loss and second-order proximity loss $\alpha$ and the embedding dimension $d$. For brevity, we report the results by link prediction task with drug dataset.


\begin{figure}
	\includegraphics[width=1.0\linewidth]{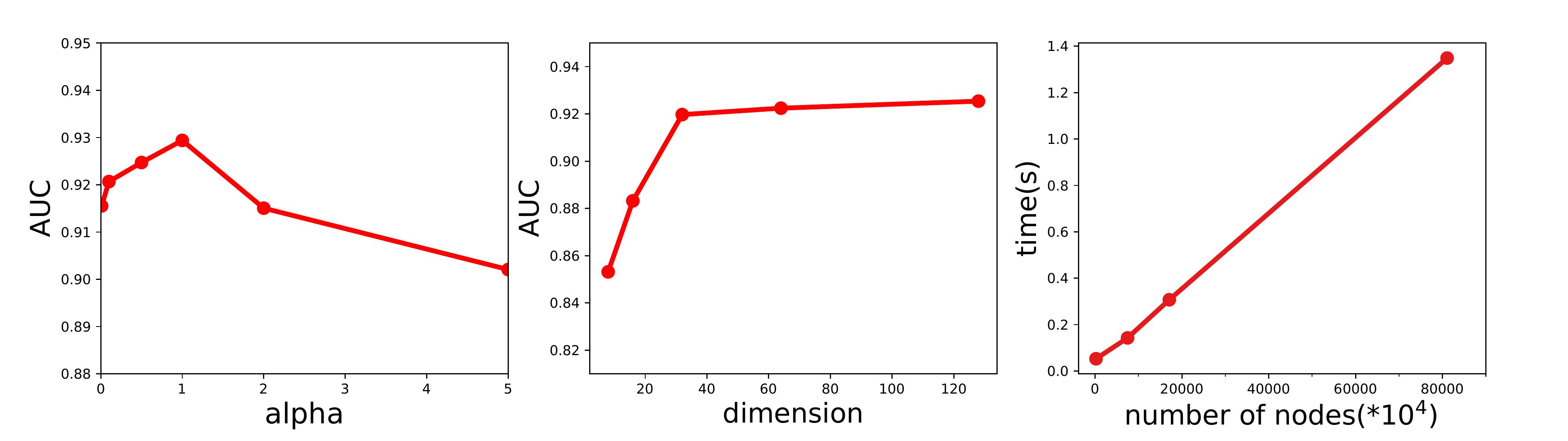}
	\caption{\small left, middle: Parameter w.r.t. embedding dimensions $d$, the value of $\alpha$; right: training time per batch w.r.t. embedding dimensions.}
	\label{fig:ps}
\end{figure}

\subsubsection{Effect of embedding dimension}
We show how the dimension of embedding space affects the performance in Figure \ref{fig:ps} left. We can see that the performance raises firstly when the number of embedding dimension increases. This is reasonable because higher embedding dimensions can embody more information of a  hyper-network. After the embedding dimension is larger than 32, the curve is relatively stable, demonstrating that our algorithm is not very sensitive to embedding dimension.

\subsubsection{Effect of parameter $\alpha$}
The parameter $\alpha$ measures the trade-off of the first-order proximity and second-order proximity of a hyper-network. We show how the values of $\alpha$ affect the performance in Figure \ref{fig:ps} middle. When $\alpha$ equals 0, only the first-order proximity is taken into account in our method. The performance with $\alpha$ between 0.1 and 2 is better than that of $\alpha = 0$,  demonstrating the importance of second-order proximity. The fact that the performance with $\alpha$ between 0.1 and 2 is better than that of $\alpha = 5$ can demonstrate the importance of the first-order proximity. In summary, both the first-order proximity and the second-order proximity are necessary for hyper-network embedding.

\subsubsection{Training time analysis}
To testify the scalability, we test training time per batch, as shown in Figure~\ref{fig:ps} right. We can observe that the training time scales linearly with the number of nodes. This results conform to the above complexity analysis and indicate the scalability of our model.

\section{Conclusion}
In this paper, we propose a novel deep model named DHNE to learn the low-dimensional representation for hyper-networks with indecomposible hyperedges. More specifically, we theoretically prove that any linear similarity metric in embedding space commonly used in existing methods cannot maintain the indecomposibility property in hyper-networks, and thus propose a new deep model to realize a non-linear tuplewise similarity function while preserving both local and global proximities in the formed embedding space. We conduct extensive experiments on four different types of hyper-networks, including a GPS network, an online social network, a drug network and a semantic network. The empirical results demonstrate that our method can significantly and consistently outperform the state-of-the-art algorithms.

\section{Acknowledgments}
This work is supported by National Program on Key Basic Research Project No. 2015CB352300, National Natural Science Foundation of China Major Project No. U1611461; National Natural Science Foundation of China No. 61772304, 61521002, 61531006, 61702296; NSF IIS-1650723 and IIS-1716432. Thanks for the research fund of Tsinghua-Tencent Joint Laboratory for Internet Innovation Technology, and the Young Elite Scientist Sponsorship Program by CAST. Peng Cui, Xiao Wang and Wenwu Zhu are the corresponding authors. All opinions, findings, conclusions and recommendations in this paper are those of the authors and do not necessarily reflect the views of the funding agencies.

\bibliographystyle{aaai}
\bibliography{ref}
\end{document}